\pdfoutput=1
\RequirePackage{ifpdf}
\ifpdf 
\documentclass[pdftex]{sigma}
\else
\documentclass{sigma}
\fi

\numberwithin{equation}{section}

\newtheorem{Theorem}{Theorem}[section]
\newtheorem*{Theorem*}{Theorem}

\newtheorem{Lemma}[Theorem]{Lemma}
\newtheorem{Proposition}[Theorem]{Proposition}
 { \theoremstyle{definition}

\newtheorem{Remark}[Theorem]{Remark} }

\begin{document}

\allowdisplaybreaks

\newcommand{\arXivNumber}{2301.09075}

\renewcommand{\PaperNumber}{048}

\FirstPageHeading

\ShortArticleName{Koenigs Theorem and Superintegrable Liouville Metrics}

\ArticleName{Koenigs Theorem and Superintegrable Liouville\\ Metrics}

\Author{Galliano VALENT}

\AuthorNameForHeading{G.~Valent}

\Address{Laboratoire de Physique Math\'ematique de Provence,\\
 12 Rue Fabrot, 13100 Aix-en-Provence, France}
\Email{\href{mailto:galliano.valent@orange.fr}{galliano.valent@orange.fr}}

\ArticleDates{Received January 26, 2023, in final form July 08, 2023; Published online July 19, 2023}

\Abstract{In a first part, we give a new proof of Koenigs theorem and, in a second part, we determine the local form of all the superintegrable Riemannian Liouville metrics as well as their global geometries.}


\Classification{32C05; 37E35; 37E99; 37J35}

\section{Introduction}
The pioneering work of Koenigs \cite{Ko} which led to a finite set of metrics with non constant curvature and with a superintegrable (SI) geodesic flow was
rediscovered and generalized in~\cite{kk1} and \cite{kk2} at the classical and at the quantum levels, leading to the determination of its geodesics~\cite{behr, Fo, Va1} and allowed impressive three dimensional generalizations thanks to ideas from conformal geometry \cite{fh}. Needless to say, this field is closely related to superintegrability, a~flourishing domain of research \cite{mpw}.

The original proof given by Koenigs is not so convincing and in \cite{kk1} and \cite{kk2} some complex coordinates changes are even needed to recover the real forms of the metrics.

The quest for a proof, in a more modern approach, was first given in \cite{bmm}: after having proved that SI metrics are metrics admitting three independent projective fields the authors proceeded to give a complete classification and description of these metrics with arbitrary signature.

Unaware of these results and starting from Matveev and Shevchishin framework \cite{ms} the complete list of Riemannian Koenigs metrics was derived in
\cite[Theorem 3]{Va1}. According to the~$y$ dependence of the quadratic integrals one has
\begin{alignat}{3}
&\mbox{hyperbolic:} \quad &&g=\frac{(a \cos x+b)}{\sin^2 x}\big({\rm d}x^2+{\rm d}y^2\big),&\nonumber\\
&\mbox{trigonometric:} \quad&&
\begin{cases}
 g_0=\big(a {\rm e}^{-x}+b {\rm e}^{-2x}\big)\big({\rm d}x^2+{\rm d}y^2\big),\vspace{1mm} \\
g_+=\dfrac{a \sinh x+b}{\cosh^2 x}\big({\rm d}x^2+{\rm d}y^2\big),\vspace{1mm} \\
g_-=\dfrac{a \cosh x+b}{\sinh^2 x}\big({\rm d}x^2+{\rm d}y^2\big),
\end{cases}&\nonumber\\
&\mbox{affine:} \quad &&
g=\frac{\big(a_2 x^2+2a_1 x+a_0\big)}{(a_2 x+a_1)^2}\big({\rm d}x^2+{\rm d}y^2\big).&\label{Kmain}
\end{alignat}
The parameters which do appear must be restricted in order to avoid a {\em constant} Gaussian curvature and to insure that the metric is indeed of Riemannian signature.

Koenigs theorem stemmed from the following problem: start from a generic Hamiltonian on some surface of revolution
\[ H=a^2(x)\big(P_x^2+P_y^2\big)
\]
find all the possible quadratic integrals
\[ {\mathcal S}=A(x,y) P_x^2+B(x,y)P_x P_y+C(x,y) P_y^2.
\]
Our first aim will be to give a direct proof of Koenigs theorem facing up to the PDEs which allow for the construction of ${\mathcal S}$.

On the way, we were also interested in another related problem: is it possible to find all Riemannian Liouville metrics which are SI? Some partial answers were already found in the literature \cite{Ma} and not surprisingly this problem is also related with Koenigs metrics.

Let us observe that this problem, in the pseudo-Riemannian setting, is much richer than in the Riemannian case. It was first considered in \cite{dy} and generalized in \cite{bmp} where all the normal forms of the pseudo-Riemannian metrics with a quadratic integral were determined, showing three possible classes whereas for the Riemannian metrics there is a single class.

The content of this article is the following: in Section~\ref{section2}, we present our direct proof of Koenigs theorem. In Section~\ref{section3}, we give the setting of SI Liouville metrics and the PDE which determines these metrics. Then in Section~\ref{section4}, solving this PDE reveals a {\em finite set} of metrics. In Section~\ref{global}, the global structure of these metrics is discussed and, in Section~\ref{section6}, we interpret our results in terms of coupling constant metamorphosis. Three appendices conclude the paper: Appendix~\ref{appendixA} which describes a non canonical metric on ${\mathbb H}^2$, Appendix~\ref{appendixB} gives a summary of the various SI Liouville metrics found and Appendix~\ref{appendixC} which relates the metrics given in \cite{bmm} and the ones given by (\ref{Kmain}).

\section{A direct proof of Koenigs theorem}\label{section2}
 The setting is the following: we have for Hamiltonian
\[ H=\Pi_x^2+a^2(x)P_y^2, \qquad \Pi_x=a(x)P_x,
\]
for which $\{H,P_y\}=0$ and we are looking for an extra quadratic integral (which should be irreducible) having the most general form
\[ {\mathcal S}=A(x,y) H+B(x,y) \Pi_x P_y+C(x,y) P_y^2.\]

\begin{Remark} We will exclude a constant Gaussian curvature $R=aa''-(a')^2$ since this implies the reducibility of ${\mathcal S}$.
\end{Remark}

\begin{Remark} An additive constant either in $A$ or in $C$ is irrelevant.\end{Remark}

Let us begin with

\begin{Lemma}\label{lemma1}
${\mathcal S}$ is an integral iff $(A, B, C)$ solve the following system of PDEs:
\[ \{H,{\mathcal S}\}=0\quad\Longleftrightarrow\quad
\begin{array}{@{}clcl}
(1)\colon & \partial_x A=0, & \qquad (2)\colon & \partial_x B=-a \partial_y A,\\[6mm]
(3)\colon & \partial_x C=-\partial_y(a B), & \qquad
(4)\colon & \partial_y C=\partial_x(a B).\end{array}
\]
\end{Lemma}

\begin{proof} It follows from an elementary computation of Poisson brackets. Let us observe that the functions $aB$ and $C$ are harmonic conjugate.
\end{proof}

\begin{Lemma}
The previous system admits a solution iff
\begin{gather}\label{eqI}
\alpha(x)\partial_y^3 A(y)+\lambda(x) \partial_y A(y)-\partial_y^2\beta(y)-\mu(x) \beta(y)=0,
\end{gather}
where
\[
\lambda=\frac{(\alpha\alpha')''}{\alpha'},\qquad \mu=\frac{\alpha'''}{\alpha'},\qquad \alpha'=a.\]
\end{Lemma}

\begin{proof} Relations $(1)$ and $(2)$ are easily integrated
\[
A=A(y),\qquad B(x,y)=-\alpha(x) \partial_y A+\beta(y).
\]
The remaining equations (3) and (4) become
\begin{equation*}
\partial_x C=\alpha\alpha'\partial_y^2A-\alpha' \partial_y\beta, \qquad
\partial_y C=-(\alpha\alpha')' \partial_yA+\alpha'' \beta.
\end{equation*}
Their integrability condition gives relation~(\ref{eqI}).
\end{proof}

Let us prove:

\begin{Proposition}\label{first}
The integrability equation~\eqref{eqI} is equivalent to
\begin{equation*}
B=-\alpha \partial_yA,\qquad \partial_y^3A+s_0 \partial_yA=0,\qquad \big(\alpha^2\big)'''-s_0 \big(\alpha^2\big)'=0, \qquad s_0\in\{0,\pm 1\}.
\end{equation*}
\end{Proposition}

\begin{proof} Let us first deduce some necessary conditions. Applying $\partial_x$ to~(\ref{eqI}) gives
\begin{gather}\label{eqII}
 \alpha' \partial_y^3A+\lambda' \partial_yA-\mu' \beta=0.
\end{gather}
Observing that $\mu'$ cannot vanish, otherwise the Gaussian curvature is constant, this relation determines $\beta(y)$ according to
\[
\beta=\frac{\alpha'}{\mu'}\partial_y^3A+\frac{\lambda'}{\mu'}\partial_yA.
\]
Applying on it $\partial_x$ leads to
\begin{gather*}
\left(\frac{\alpha'}{\mu'}\right)' \partial_y^3A+\left(\frac{\lambda'}{\mu'}\right)' \partial_yA=0,
\end{gather*}
from which we deduce the separation relations
\begin{equation}\label{sep}
\partial_y^3A+s_0\partial_yA=0,\qquad s_0\in{\mathbb R},\qquad \left(\frac{\lambda'-s_0\alpha'}{\mu'}\right)'=0.
\end{equation}
By a scaling of $y$, we can restrict $s_0$ to be in $\{0,\pm 1\}.$

The relations obtained in (\ref{sep}) are necessary. To lift them up to sufficient conditions we must discuss backwards successively the relations~(\ref{eqI}) and~(\ref{eqII}).

Integration of the right-hand relation in (\ref{sep}) produces
\begin{equation}\label{int2}
\lambda'-s_0\alpha'=l \mu'\quad\Longrightarrow\quad \lambda-s_0\alpha=l \mu+m, \qquad (l,m)\in{\mathbb R}^2.
\end{equation}
Since $\mu'$ cannot vanish, relation~(\ref{eqII}) determines $\beta$:
\[
\beta=l \partial_yA \quad\Longrightarrow\quad B=-(\alpha-l)\partial_yA.
\]

Eventually relation~(\ref{eqI}) gives $m=-s_0 l$. Hence the right-hand relation in (\ref{int2}) leads to the following ODE
\[ \big[(\alpha-l)^2\big]'''-s_0\big[(\alpha-l)^2\big]'=0,\]
showing that we can take $l=0$ since $a=\alpha'=(\alpha-l)'$.

It follows that $\beta(y)=0$ and $B=-\alpha \partial_yA$, ending the proof.
\end{proof}

So we conclude, in agreement with \cite{Ko} and \cite{Va1}:

\begin{Theorem}[Koenigs]\label{theorem1}
Recalling that $H=(\alpha')^2\big(P_x^2+P_y^2\big)$, according to the values of $s_0$ we have:
\begin{alignat*}{4}
&s_0= 0, \qquad && A=k_1 y+k_2 \frac{y^2}{2}, \qquad && \alpha^2=a_0+a_1x+a_2x^2,&\nonumber\\
&s_0=+1, \qquad && A=k_1 \sin y+k_2 \cos y, \qquad &&\alpha^2=a_0+a_1\sinh x+a_2\cosh x,&\nonumber\\
& s_0=-1, \qquad && A=k_1 \sinh y+k_2 \cosh y, \qquad && \alpha^2=a_0+a_1\sin x+a_2\cos x,&
\end{alignat*}
while ${\mathcal S}$ is given by
\begin{alignat}{3}
&s_0=0, \qquad && {\mathcal S}=A H-\frac{\big(\alpha^2\big)'}{2} (\partial_y A) P_x P_y+\frac 12 \big(\alpha^2 \partial_y^2 A-\big(\alpha^2\big)''A\big) P_y^2,&\nonumber \\
&s_0=\pm 1, \qquad &&
{\mathcal S}=A H-\frac{\big(\alpha^2\big)'}{2} (\partial_y A) P_x P_y-\frac{s_0}{2}\big(\alpha^2-a_0\big) A P_y^2.&\label{integrals}
\end{alignat}
\end{Theorem}

\begin{proof} Proposition~\ref{first} gives
\[ aB=-\frac{\big(\alpha^2\big)'}{2} \partial_yA \]
and $C$ (see Lemma~\ref{lemma1}) has to be determined from
\[ \partial_x C=\frac{\alpha^2}{2} \partial^2_yA,\qquad \partial_y C=-\frac 12\big(\alpha^2\big)'' \partial_yA. \]
Elementary computations give $C$ in each case.
\end{proof}

\begin{Remark} Let us notice the striking balance between trigonometric and hyperbolic lines in the differential equations obtained in Proposition \ref{first}:
\[ \partial_y^3A+s_0 \partial_yA=0,\qquad \big(\alpha^2\big)'''-s_0 \big(\alpha^2\big)'=0.\]\end{Remark}
\begin{Remark} One can check that $aB$ and $C$ are indeed harmonic conjugate.\end{Remark}

\begin{Remark} From Theorem~\ref{theorem1}, we observe that we can define ${\mathcal S}=k_1 {\mathcal S}_1+k_2 {\mathcal S}_2$, where these integrals are easily deduced from (\ref{integrals}).
It follows that the linear span of the quadratic integrals is indeed four dimensional, with
\[ H,\quad P_y^2,\quad {\mathcal S}_1,\quad {\mathcal S}_2.\]
\end{Remark}

\begin{Remark} However a well known theorem (for a proof see \cite[p.~461]{vds}) states that in dimension 2 there may be at most three {\em algebraically independent} integrals including the Hamiltonian. Indeed there is a quadratic relation between ${\mathcal S}_1$ and ${\mathcal S}_2$ given in \cite{kk1} and \cite{kk2} which does reduce the number of algebraically independent integrals to three.
\end{Remark}

\begin{Remark} The global properties of these SI metrics, their geodesics and their quantum structure were discussed in \cite{behr} and in \cite{Va1}.
\end{Remark}

For further reference, let us mention three metrics of interest. The first one is
\begin{equation}\label{met1}
s_0=0,\qquad \alpha=2\sqrt{x}\colon\ g=x\big({\rm d}x^2+{\rm d}y^2\big).
\end{equation}
A second one is
\begin{equation}\label{met2}
s_0=0,\qquad \alpha=\sqrt{s+x^2}\colon\ g=\big(s+x^2\big)\frac{\big({\rm d}x^2+{\rm d}y^2\big)}{x^2}.
\end{equation}
The last one is
\[s_0=-1, \qquad \alpha=\frac 4{s}\big(\sqrt{1+s {\rm e}^x}-1\big)\colon\ g=\big(s {\rm e}^{-x}+{\rm e}^{-2x}\big)\frac{\big({\rm d}x^2+{\rm d}y^2\big)}{4}.
\]
The change of coordinates $X={\rm e}^{-x/2}\cos(y/2)$ and $Y= {\rm e}^{-x/2}\sin(y/2)$ leads to
\begin{equation}\label{met3}
g=\big(s+X^2+Y^2\big)\big({\rm d}X^2+{\rm d}Y^2\big).\end{equation}

Let us consider now the second item: SI Liouville metrics.

\section{The setting}\label{section3}
The metric on a Liouville surface is defined locally as
\[
g=(f(x)+g(y))\big({\rm d}x^2+{\rm d}y^2\big),
\]
and its geodesic flow, generated by the Hamiltonian
\[
H=\frac{P_x^2+P_y^2}{f(x)+g(y)},
\]
does exhibit the quadratic integral
\[
Q=\frac{g(y)P_x^2-f(x)P_y^2}{f(x)+g(y)}=P_x^2-f(x) H=g(y)H-P_y^2,\qquad\{H,Q\}=0.
\]
To reach superintegrability, we need an extra integral which can be written without loss of generality
\[
{\mathcal S}=A(x,y) H+B(x,y) P_x P_y+C(x,y) P_y^2.
\]

\begin{Remark} We will exclude from our analysis the metrics $g$ of constant curvature since in this case all the integrals become reducible.
\end{Remark}

\begin{Remark} Constants either in $A$ or in $C$ are irrelevant.\end{Remark}

\subsection{PDE system and a master equation}
Let us start with

\begin{Proposition}\label{prop2}
The conservation relation $\{H,{\mathcal S}\}=0$ is
equivalent to the following set of PDE:
\begin{itemize}\itemsep=0pt
 \item[$(1)$] $ \partial_x A+\frac{\dot{g}}{2} B=0,$
\item[$(2)$] $ \partial_y A+(f+g) \partial_x B+\frac{f'}{2} B+\dot{g} C=0,$
\item[$(3)$] $\partial_x A +(f+g)(\partial_y B+\partial_x C)+\frac{\dot{g}}{2} B=0,$
\item[$(4)$] $\partial_y A+(f+g) \partial_y C+\frac{f'}{2} B+\dot{g} C=0.$
\end{itemize}
\end{Proposition}

\begin{proof} Elementary calculation.
\end{proof}

\subsection{Discussion}
Computing relations (1)--(3) and (2)--(4) gives
\begin{equation}\label{H}
\partial_x C=-\partial_y B, \qquad \partial_y C=\partial_x B,
\end{equation}
implying that these two functions are harmonic conjugate while $A$ is determined by
\[ \partial_x A=-\frac{\dot{g}}{2} B, \qquad \partial_y A=-f \partial_y C-\frac{f'}{2} B-\partial_y(g C).\]
So, defining $\ {\mathcal A}=A+g C,\ $ we end up with
\begin{equation}\label{S}
\begin{cases}
\partial_x{\mathcal A}=-g \partial_yB-\dfrac{\dot{g}}{2} B,\vspace{1mm}\\
\partial_y{\mathcal A}=-f \partial_x B-\dfrac{f'}{2} B,
\end{cases}
\quad\Longrightarrow\quad {\mathcal S}={\mathcal A} H+B P_x P_y-C Q.
\end{equation}
The integrability condition of this system is therefore
\begin{equation}\label{int}
B f''+3 B' f'+2B'' f=B \ddot{g}+3 \dot{B} \dot{g}+2\ddot{B} g.
\end{equation}
This is the master equation to be solved.

Having observed in (\ref{H}) that $B$ and $C$ are harmonic conjugate, let us define
\[
(x+{\rm i}y)^n={\mathcal R}_n(x,y)+{\rm i}{\mathcal I}_n(x,y).
\]
Since equation (\ref{int}) is linear, we can consider separately the following four cases:
\begin{alignat*}{3}
&{\rm I}\colon\hphantom{{\rm II}} \ B={\mathcal R}_{2n}(x,y),\qquad && {\rm II}\colon \ \ B={\mathcal I}_{2n+1}(x,y),& \\
& {\rm III}\colon \
B={\mathcal R}_{2n+1}(x,y), \qquad& & {\rm IV}\colon \ B={\mathcal I}_{2n}(x,y).&
\end{alignat*}

\begin{Remark} Despite the usefulness of ${\mathcal A}$, and since $Q=g(y)H-P_y^2$, we will write the integral~${\mathcal S}$ in its initial form
\[ {\mathcal S}=A(x,y)H+B(x,y) P_x P_y+C(x,y) P_y^2.\]
\end{Remark}

\section{Discussion of the four cases}\label{section4}

\subsection{Case I}\label{I}
Let us prove
\begin{Proposition}\label{p3} For $B$ given by
\[
B={\mathcal R}_{2n}(x,y)=\sum_{l=0}^n {2n \choose 2l}(-1)^{n-l} x^{2l} y^{2(n-l)},\qquad n\geq 0,
\]
the integrability constraint \eqref{int} never gives a Liouville metric except for $n=0$ for which we have
\begin{equation}\label{fcfg}
f(x)=\mu x^2+2a_1 x+a_0,\qquad g(y)=\mu y^2+2b_1 y+b_0
\end{equation}
and
\begin{equation}\label{fcAS}
{\mathcal A}=-\mu x y-b_1 x-a_1 y,\qquad {\mathcal S}={\mathcal A} H+P_x P_y.
\end{equation}
\end{Proposition}

\begin{proof} For $n=0$, we have $B=1$ and $C=0$. Relation (\ref{int}) becomes $f''=\ddot{g}=2\mu$ with
$\mu\in{\mathbb R}$ which gives (\ref{fcfg}). Integrating (\ref{S}) gives (\ref{fcAS}).

Let us consider $n\geq 1$. We will call $S_1$ (resp.\ $S_2$) the left-hand side (resp.\ the right-hand side) of the relation (\ref{int}). Positing $C^n_l=(-1)^{n-l}{2n \choose 2l}$ and $\delta_x=x \partial_x$ and $\delta_y=y \partial_y$, we have
\[
S_1=\sum_{l=0}^n C^n_l x^{2(l-1)} F_l(x) y^{2(n-l)},\qquad F_l(x)=(\delta_x+2l-1)(\delta_x+4l)f(x),
\]
while
\[
S_2=\sum_{l=0}^n C^n_l x^{2l} y^{2(n-l-1)}G_{n-l}(y),\qquad G_{\nu}(y)=(\delta_y+2\nu-1)(\delta_y+4\nu)g(y).
\]

The leading terms are
\[
S_1=C^n_0 f''(x) y^{2n}+\cdots,\qquad
S_2=C^n_n x^{2n} \ddot{g}(y)+\cdots.
\]
Acting with $\partial_x^{2n} \partial_y^{2n}$ on the relation $S_1=S_2$ gives the separation relation
\[
\partial_x^{2(n+1)}f(x)=(-1)^n\partial_y^{2(n+1)}g(y),
\]
which leads to polynomials for $f$ and $g$
\[
f(x)=\sum_{k=0}^{2(n+1)}a_k x^k,\qquad g(y)=\sum_{k=0}^{2(n+1)}b_k y^k,\qquad b_{2n+2}=(-1)^n a_{2n+2}.
\]
Now we can compute
\begin{gather*}
S_1=\sum_{k=0}^{2(n+1)}\sum_{l=0}^n C^n_l a_k (k+2l-1)(k+4l) x^{k+2(l-1)} y^{2(n-l)},\\
 S_2=\sum_{k=0}^{2(n+1)}\sum_{l=0}^n C^n_l b_k (k+2(n-l)-1)(k+4(n-l)) x^{2l} y^{k+2(n-l)}.
\end{gather*}
The relation $S_1=S_2$ implies that $k$ can take only even values, giving
\begin{gather*}
\widetilde{S}_1=\frac{S_1}{2}=\sum_{k=0}^{n+1}\sum_{l=0}^n C^n_l a_{2k} (2k+2l-1)(k+2l) x^{2(k+l-1)} y^{2(n-l)},\\
\widetilde{S}_2=\frac{S_2}{2}=\sum_{k=0}^{n+1}\sum_{l=0}^n C^n_l b_{2k} (2(k+n-l)-1)(k+2(n-l)) x^{2l} y^{2(k+n-l-1)}.
\end{gather*}
The first sum can be written
\begin{equation}\label{fsum}
\widetilde{S}_1=a_0 A_0+\sum_{k=1}^{n+1}\sum_{l=0}^n C^n_l a_{2k} (2k+2l-1)(k+2l) x^{2(k+l-1)} y^{2(n-l)},
\end{equation}
where
\[ A_0=\sum_{l=0}^{n-1} C^n_{l+1}(2l+1)(2l+2) x^{2l} y^{2(n-l-1)}.
\]
In \eqref{fsum}, let us operate the change of summation index $l \to L=l+k-1$ giving
\begin{equation}\label{S11}
\sum_{k=1}^{n+1}\sum_{L=k-1}^{n+k-1} C^n_{L-k+1} a_{2k} (2L+1)(2L-k+2) x^{2L} y^{2(k+n-L-1)}.
\end{equation}
Exchanging the summations in (\ref{S11}) leads to
\begin{align*}
\widetilde{S}_1=a_0 A_0&+\sum_{L=0}^n(2L+1)x^{2L}\sum_{k=1}^{L+1} C^n_{L-k+1}(2L-k+2) a_{2k} y^{2(k+n-L-1)}\\
&{}+\sum_{L=n+1}^{2n}(2L+1)x^{2L}\sum_{k=L-n+1}^{n+1} C^n_{L-k+1} (2L-k+2) a_{2k} y^{2(k+n-L-1)}.\end{align*}
Comparing with $\widetilde{S}_2$ shows that the last sum must vanish because $L\geq n+1$ and this entails $a_{2s}=0$ for $s=2,3,\ldots,n+1$, reducing $\widetilde{S}_1$ to the simple form
\[
\widetilde{S}_1=a_0 A_0+a_2 A_2,\qquad A_2=\sum_{l=0}^n C^n_l (2l+1)^2 x^{2l} y^{2(n-l)}.
\]
Let us consider $\widetilde{S}_2$, which can be written
\begin{equation*}
\widetilde{S}_2=b_0 B_0+b_2 B_2+\sum_{l=0}^n C^n_l x^{2l}\sum_{k=2}^{n+1} b_{2k}(2(k+n-l)-1)(k+2(n-l)) y^{2(n-l+k-1)},
\end{equation*}
where
\begin{gather*}
B_0=\sum_{l=0}^{n-1} C^n_l(2(n-l)-1)2(n-l) x^{2l} y^{2(n-l-1)},\qquad
B_2=\sum_{l=0}^n C^n_l (2(n-l)+1)^2 x^{2l} y^{2(n-l)}.
\end{gather*}
Comparing $\widetilde{S}_2$ with $\widetilde{S}_1$ implies $b_{2s}=0$ for $s=2,3,\ldots,n+1$ and we are left with
\[ \widetilde{S}_2=b_0 B_0+b_2 B_2.\]
The relation
\[ (2l+1)(2l+2) C^n_{l+1}=-(2(n-l)-1)(2n-2l) C^n_l\quad \Longrightarrow\quad B_0=-A_0,
\]
so we end up with
\[ \widetilde{S}_1-\widetilde{S}_2=(a_0+b_0)A_0+a_2 A_2-b_2 B_2=0,
\]
or more explicitly
\begin{gather*} \widetilde{S}_1-\widetilde{S}_2=(a_0+b_0)\big(2C^n_1 y^{2(n-1)}+\cdots +(2n-1)2n C^n_n x^{2(n-1)}\big) \\
\phantom{\widetilde{S}_1-\widetilde{S}_2=}{}+a_2\big(C^n_0 y^{2n}+\cdots+(2n+1)^2C^n_n x^{2n}\big)
-b_2\big((2n+1)^2 C^n_0 y^{2n}+\cdots+C^n_n x^{2n}\big)=0.
\end{gather*}
It follows that $a_0+b_0=0$ and $a_2=b_2=0$, hence $f(x)+g(y)\equiv 0$.
\end{proof}

Let us consider the second case.

\subsection{Case II}\label{II}
Let us prove
\begin{Proposition}\label{p4} For $B$ given by
\[ B={\mathcal I}_{2n+1}(x,y)=\sum_{l=0}^{n} {2n+1 \choose 2l}(-1)^{n-l} x^{2l} y^{2(n-l)+1}, \qquad n\geq 0,
\]
the integrability constraint \eqref{int} never gives a Liouville metric except if $n=0$ for which we have%
\begin{equation}\label{scfg}
f(x)=\mu x^2+2a_1 x+a_0,\qquad g(y)=\frac{\mu}{4}y^2+\frac{b_{-2}}{y^2}+b_0,
\end{equation}
and
\begin{equation}\label{scAS}
{\mathcal A}=-\frac{\mu}{2}xy^2-\frac{a_1}{2}y^2-b_0 x\quad\Longrightarrow\quad {\mathcal S}={\mathcal A} H+y P_x P_y+xQ.
\end{equation}
\end{Proposition}

\begin{proof} For $n=0$, we have $B=y$ and $C=-x$. The relation (\ref{int}) becomes
\[ f''=\ddot{g}+\frac 3y\dot{g}=2\mu,\qquad \mu\in {\mathbb R},\]
which implies (\ref{scfg}). Integrating (\ref{S}) gives (\ref{scAS}).

Let us consider $n\geq 1$. Positing this time $ C^n_l=(-1)^{n-l}{2n+1 \choose 2l},$ we have
\begin{equation*}
S_1=\sum_{l=0}^{n} C^n_l x^{2(l-1)} F_l(x) y^{2(n-l)+1},\qquad
F_l(x)=(\delta_x+2l-1)(\delta_x+4l)f(x),
\end{equation*}
and
\begin{equation*}
S_2=\sum_{l=0}^{n-1} C^n_l x^{2l} y^{2(n-l)-1}G_{n-l}(y),\qquad G_{\nu}(y)=(\delta_y+2\nu)(\delta_y+4\nu+2)g(y).
\end{equation*}
The leading terms are
\[ \frac{S_1}{y}=C^n_0 f''(x) y^{2n}+\cdots,\qquad \frac{S_2}{y}=C^n_n x^{2n}\left(\ddot{g}+\frac 3y\dot{g}\right)+\cdots,
\]
so acting with $\partial_x^{2n}\partial_y^{2n}$ gives the separation relation
\[ \partial_x^{2(n+1)}f(x)=(-1)^n(2n+1)\partial_y^{2n} \left(\ddot{g}+\frac 3y\dot{g}\right),
\]
which are easily integrated into
\[ f(x)=\sum_{k=0}^{2(n+1)}a_k x^k,\qquad g(y)=\sum_{k=-2,k\neq \pm 1}^{2(n+1)}b_k y^k,\qquad b_{2(n+1)}=\frac{(-1)^n}{2(n+2)} a_{2(n+1)}.
\]
Similarly to Case~I, we obtain
\begin{gather*}
\begin{split}
& \widetilde{S}_1=\sum_{l=0}^n\sum_{k=0}^{n+1} C^n_l
(2(k+l)-1)(k+2l) a_{2k} x^{2(k+l-1)} y^{2(n-l)+1},\\
& \widetilde{S}_2=\sum_{l=0}^n\sum_{k=-1}^{n+1} C^n_l 2(k+n-l)(k+2(n-l)+1) b_{2k} x^{2l} y^{2(k+n-l)-1}.
\end{split}
\end{gather*}
Following the same lines as in the proof of Proposition~\ref{first}, we get
\begin{gather*}
\widetilde{S}_1=a_0 A_0+a_2 A_2,\\
A_0=\sum_{l=0}^{n-1} C^n_{l+1} (2l+1)(2l+2) x^{2l} y^{2(n-l)-1},\qquad
A_2=\sum_{l=0}^n C^n_l (2l+1)^2 x^{2l} y^{2(n-l)+1}.
\end{gather*}
Let us consider now $\widetilde{S}_2$. It can be expanded
\begin{gather*}
\widetilde{S}_2=b_{-2} B_{-2}+b_0 B_0+b_2 B_2 +\sum_{l=0}^n\sum_{k=2}^{n+1} C^n_l 2(k+n-l)(k+2(n-l)+1) b_{2k} x^{2l} y^{2(k+n-l)-1},\end{gather*}
where
\[ B_{-2}=\sum_{l=0}^n C^n_l 2(n-l-1)2(n-l) x^{2l} y^{2(n-l)-3}.
\]
Comparing with $\widetilde{S}_1$ implies
\[ b_{-2}=0,\qquad b_4=b_6=\cdots=b_{2(n+1)}=0
\]
and we are left with
\begin{gather*}
 \widetilde{S}_2=b_0 B_0+b_2 B_2,\\
 B_0=\sum_{l=0}^{n-1} C^n_l 2(n-l)(2(n-l)+1)x^{2l} y^{2(n-l)-1},\\
 B_2=\sum_{l=0}^n C^n_l [2(n-l+1)]^2 x^{2l} y^{2(n-l)+1}.\end{gather*}
Here also we have $B_0=-A_0$ and the relation $\widetilde{S}_1=\widetilde{S}_2$ reduces to
\[ (a_0+b_0)A_0+a_2A_2-b_2B_2=0. \]
Since we have
\begin{align*}
A_0&=2C^n_1 y^{2n-1}+\cdots+(2n-1)2n C^n_n x^{2(n-1)}y,\\
A_2&=C^n_0 y^{2n+1}+\cdots+(2n+1)^2C^n_n x^{2n} y, \\
B_2&=4(n+1)^2 C^n_0 y^{2n+1}+\cdots
+4C^n_n x^{2n} y,
\end{align*}
it follows that for $n\geq 1$ we have $a_0+b_0=0$ and $a_2=b_2=0$ hence $f(x)+g(y)\equiv 0$.
\end{proof}

Let us consider the third case.

\subsection{Case III}\label{III}
Let us prove
\begin{Proposition}\label{p5} For $B$ given by
\[ B={\mathcal R}_{2n+1}(x,y)=\sum_{l=0}^n(-1)^{n-l}{2n+1 \choose 2l+1} x^{2l+1} y^{2(n-l)},\qquad n\geq 0,
\]
the integrability constraint \eqref{int} never gives a Liouville metric except for $n=0$ for which we have
\begin{equation}\label{tcfg}
f(x)=\frac{\mu}{4} x^2+\frac{a_{-2}}{x^2}+a_0,\qquad g(y)=\mu y^2+2b_1 y+b_0,
\end{equation}
and
\begin{equation}\label{tcAS}
{\mathcal A}=-\frac{\mu}{2} x^2 y-\frac{b_1}{2}x^2-a_0 x \quad \Longrightarrow\quad {\mathcal S}={\mathcal A} H+x P_x P_y-y Q.\end{equation}
\end{Proposition}

\begin{proof} For $n=0$, we have $B=x$ and $C=y$. The relation (\ref{int}) becomes
\[ f''+\frac 3xf'=\ddot{g}=2\mu,\]
which implies (\ref{tcfg}). Integrating (\ref{S}) gives (\ref{tcAS}).

Let us posit $ C^n_l=(-1)^{n-l}{2n+1 \choose 2l+1}$. This time we have
\[ S_1=\sum_{l=0}^n C^n_l x^{2l-1} F_l(x) y^{2(n-l)},\qquad F_l(x)=(\delta_x+2l)(\delta_x+4l+2)f(x),
\]
and
\[ S_2=\sum_{l=0}^n C^n_lx^{2l+1} y^{2(n-l-1)} G_{n-l}(y),\qquad G_{\nu}(y)=(\delta_y+2\nu-1)(\delta_y+4\nu)g(y).
\]
The leading terms are
\[ \frac{S_1}{x}=C^n_0\left(f''+\frac 3x f'\right)y^{2n}+\cdots,\qquad \frac{S_2}{x}=C^n_n x^{2n} \ddot{g}+\cdots,
\]
so acting with $\partial_x^{2n} \partial_y^{2n}$ we get the separation relation
\[ \partial_x^{2n}\left(f''+\frac 3x f'\right)=\frac{(-1)^n}{2n+1} \partial_y^{2(n+1)}g,
\]
which implies
\[ f(x)=\sum_{k-2,k\neq\pm 1}^{2(n+1)}a_k x^k,\qquad g(y)=\sum_{k=0}^{2(n+1)}b_ky^k,\qquad b_{2(n+1)}=(-1)^n 2(n+2) a_{2(n+1)}.
\]
As in the previous cases, we get
\begin{align*}
\widetilde{S}_1&=\sum_{l=0}^n C^n_l y^{2(n-l)}\sum_{k=-1}^{n+1}2(k+l)(k+2l+1)a_{2k} x^{2(k+l)-1}, \\
\widetilde{S}_2&=\sum_{l=0}^n C^n_l x^{2l+1}\sum_{k=0}^{n+1}(2(k+n-l)-1)(k+2(n-l))b_{2k} y^{2(k+n-l-1)}.
\end{align*}
By an analysis similar to the proof of Proposition~\ref{prop2}, we have
\[ \widetilde{S}_1=a_{-2} A_{-2}+a_0 A_0 +a_2 A_2,
\]
with
\[ n=1\colon\ A_{-2}=0,\qquad n\geq 2\colon\ A_{-2}=2\sum_{l=0}^{n-2} C^n_{l+2}(2l+2)(2l+4) x^{2l+1}y^{2(n-l-2)}
\]
and
\[ A_0=2\sum_{l=0}^{n-1}C^n_{l+1}(2l+2)(2l+3) x^{2l+1}y^{2(n-l)-1},\qquad A_2=2\sum_{l=0}^n C^n_l (2l+2)^2 x^{2l+1}y^{2(n-l)}.
\]
Since the powers of $y$ appearing in $A_{-2}$ never appear in $\widetilde{S}_2$ it follows that for $n\geq 2$ we have~${a_{-2}=0}$ and $\widetilde{S}_1=a_0 A_0+a_2 A_2.$

By an argument similar to the proof of Proposition~\ref{first}, we have
\[ \widetilde{S}_2=b_0 B_0+b_2 B_2,
\]
where
\begin{gather*}
B_0=\sum_{l=0}^{n-1} C^n_l (2(n-l)-1)2(n-l) x^{2l+1} y^{2(n-l-1)},\\
B_2=\sum_{l=0}^n C^n_l (2(n-l)+1)^2 x^{2l+1} y^{2(n-l)}.\end{gather*}
Here also the relation $B_0=-A_0$ holds. So we have
\[
\widetilde{S}_1-\widetilde{S}_2=0\quad\Longleftrightarrow\quad (a_0+b_0)A_0+a_2 A_2-b_2 B_2=0.
\]
The relations
\begin{align*}
 A_0&=12 C^n_1 xy^{2n-1}+\cdots +4n(2n+1)C^n_n x^{2n-1}y, \\
A_2&=8C^n_0 xy^{2n}+\cdots+4(n+1)^2 C^n_n x^{2n+1}, \\
B_2&=(2n+1)^2C^n_0 xy^{2n}+\cdots+C^n_n x^{2n+1}
\end{align*}
imply, for $n\geq 1$, that $a_0+b_0=0$ and $a_2=b_2=0$ hence $f(x)+g(y)\equiv 0$.
\end{proof}

\begin{Remark}\label{remark4.4} One should observe that Case~III is not really different from Case~II since, if we define
\[ f(x)=\mu x^2+2a_1 x+a_0, \qquad g(y)=\frac{\mu}{4}y^2+\frac{b_{-2}}{y^2}+b_0,\]
the Hamiltonians are respectively
\[ H_{\rm II}=\frac{P_x^2+P_y^2}{f(x)+g(y)} \qquad \mbox{and} \qquad H_{\rm III}=\frac{P_x^2+P_y^2}{g(x)+f(y)},\]
and therefore they cannot be considered as different since they are related by the sub\-sti\-tu\-tion~$x \leftrightarrow y$.
\end{Remark}

Let us consider the last case.

\subsection{Case IV}\label{IV}
Let us prove
\begin{Proposition}\label{p6} For $B$ given by
\[ B={\mathcal I}_{2n}(x,y)=\sum_{l=0}^{n-1}(-1)^{n-l}{2n \choose 2l+1} x^{2l+1} y^{2(n-l)-1},\qquad n\geq 1,
\]
the integrability constraint \eqref{int} does not give a Liouville metric except for $n=1$ and $n=2$. For $n=1$, we have
\begin{equation}\label{fg4}
f(x)=\mu x^2+\frac{a_{-2}}{x^2}+a_0,\qquad g(y)=\mu y^2+\frac{b_{-2}}{y^2}+b_0,
\end{equation}
and
\begin{equation}\label{ex4}
{\mathcal A}=-\mu x^2y^2-\frac{b_0}{2} x^2-\frac{a_0}{2} y^2,\qquad {\mathcal S}={\mathcal A} H+xy P_x P_y+\frac 12\big(x^2-y^2\big) Q,
\end{equation}
while for $n=2$ we have
\[ f(x)=\frac a{x^2},\qquad g(y)=\frac b{y^2},\]
and
\[ {\mathcal S}=4\big(ay^2-bx^2\big)H-4xy\big(x^2-y^2\big) P_x P_y-
\big(x^4-6x^2y^2+y^4\big)Q.\]
\end{Proposition}

\begin{proof} For $n=1$, let us take $B=xy$ and $C=-\frac 12\big(x^2-y^2\big)$, in which case the relation (\ref{int}) becomes
\[ f''+\frac 3x f'=\ddot{g}+\frac 3y \dot{g}=8\mu,
\]
leading to (\ref{fg4}). Integrating (\ref{S}) gives (\ref{ex4}).
Let us consider now $n\geq 2$. Positing this time $C^n_l=(-1)^{n-l}{2n \choose 2l+1}$, we have
\[ S_1=\sum_{l=0}^{n-1} C^n_l x^{2l-1} F_l(x) y^{2(n-l)-1},\qquad
F_l(x)=(\delta_x+2l)(\delta_x+4l+2)f(x),
\]
while
\[ S_2=\sum_{l=0}^{n-1} C^n_l x^{2l+1} y^{2(n-l-1)-1}G_{n-l}(y),\qquad G_{\nu}(y)=(\delta_y+2(\nu-1))(\delta_y+4\nu-2)g(y).
\]

The higher order terms are given by
\[ S_1=C^n_0 (xf''+3f') y^{2n-1}+\cdots,\qquad
S_2=C^n_{n-1} x^{2n-1}(y\ddot{g}+3\dot{g})+\cdots.
\]
So, acting with $\partial_x^{2n-1} \partial_y^{2n-1}$ on the relation $S_1=S_2$, gives the separation relation
\[ \partial_x^{2n-1}(xf''+3f')=(-1)^{n-1}\partial_y^{2n-1}(y\ddot{g}+3\dot{g}),
\]
which implies
\[ f(x)=\sum_{k=-2,k\neq \pm 1}^{2n}a_k x^k,\qquad g(y)=\sum_{k=-2, k\neq \pm 1}^{2n}b_k y^k,\qquad b_{2n}=(-1)^{n-1} a_{2n}.
\]

It follows that
\begin{align*}
\widetilde{S}_1&=\sum_{l=0}^{n-1} C^n_l y^{2(n-l)-1}\sum_{k=-1}^{n} 2(k+l)(k+2l+1) a_{2k} x^{2(k+l)-1},\\
\widetilde{S}_2&=\sum_{l=0}^{n-1} C^n_l x^{2l+1}\sum_{k=-1}^{n} 2(k+n-l-1)(k+2(n-l)-1) b_{2k} y^{2(k+n-l-1)-1}.
\end{align*}
As in the previous propositions, one can prove that
\begin{equation}\label{sumS1}
\widetilde{S}_1=a_{-2} A_{-2}+a_0 A_0+a_2 A_2,
\end{equation}
where
\[ n=2\colon\ A_{-2}=0,\qquad n\geq 3\colon\ A_{-2}=\sum_{l=0}^{n-3}C^n_{l+2} (2l+2)(2l+4) x^{2l+1} y^{2(n-l-2)-1},
\]
and
\begin{gather*}
A_0=\sum_{l=0}^{n-2} C^n_{l+1} (2l+2)(2l+3) x^{2l
+1} y^{2(n-l-1)-1},\\
A_2=\sum_{l=0}^{n-1}C^n_l (2l+2)^2 x^{2l+1} y^{2(n-l)-1}.\end{gather*}
Similarly, one can show that $\widetilde{S}_2$ can be written
\[ \widetilde{S}_2=b_2 B_{-2}+b_0 B_0+b_2 B_2,
\]
where
\[ n=2\colon\ B_{-2}=0,\qquad n\geq 3\colon\ B_{-2}=\sum_{l=0}^{n-3} C^n_l 2(n-l-2)2(n-l-1) x^{2l+1} y^{2(n-l-2)-1},\]
and
\begin{gather*}
B_0=\sum_{l=0}^{n-2} C^n_l 2(n-l-1)(2(n-l)-1) x^{2l+1} y^{2(n-l-1)-1},\\
B_2=\sum_{l=0}^{n-1} C^n_l (2(n-l))^2 x^{2l+1} y^{2(n-l)-1}.\end{gather*}
Since the relation $B_0=-A_0$ remains valid the last step in the proof requires solving
\begin{equation}\label{eqfin}
\widetilde{S}_1=\widetilde{S}_2\quad\Longleftrightarrow\quad a_{-2} A_{-2}-b_{-2} B_{-2}+(a_0+b_0) A_0+a_2 A_2-b_2 B_2=0.\end{equation}

At this stage we must discuss separately the cases $n=2$ and $n=3$ before the general case~$n\geq 4$.

\subsubsection[First case: n=2]{First case: $\boldsymbol {n=2}$}\label{7p1}
In this case, we have $A_{-2}=B_{-2}=0$ and the previous relation reduces to
\[ -3(a_0+b_0)xy+2a_2\big(xy^3-8x^3y\big)-2b_2\big(xy^3-x^3y\big)=0,
\]
implying $a_2=b_2=0$ and $a_0+b_0=0$ but no constraint on $a_{-2}=a$, $b_{-2}=b$. So we have
\[ f(x)+g(y)=\frac a{x^2}+\frac b{y^2}, \]
and integrating (\ref{S}) gives
\begin{equation}\label{casn2}
{\mathcal S}=4\big(ay^2-bx^2\big)H-4xy\big(x^2-y^2\big) P_x P_y-
\big(x^4-6x^2y^2+y^4\big)Q.
\end{equation}

\subsubsection[Second case: n=3]{Second case: $\boldsymbol{ n=3}$}
Comparing, in equation (\ref{eqfin}), the $y$ dependence of the various bivariate polynomials shows that this equation reduces to
\[ a_{-2} A_{-2}-b_{-2} B_{-2}=0,\qquad (a_0+b_0)A_0=0,\qquad a_2 A_2-b_2 A_2=0.\]
Taking into account that
\[ A_{-2}=B_{-2}=-48xy\quad\Longrightarrow\quad b_{-2}=a_{-2}=a.
\]
Since $A_0=30xy\big(3y^2-4x^2\big)$ it follows that $a_0+b_0=0$.

Observing that
\[ A_2=4\big(C^3_0 xy^3+4C^3_1 x^3y^3+9C^3_2 x^5y\big),\qquad B_2=4\big(9C^3_0 xy^3+4C^3_1 x^3y^3+C^3_2 x^5y\big),
\]
implies $a_2=b_2=0$. Setting $a=1$, we have
\[ f(x)+g(y)=\frac 1{x^2}+\frac 1{y^2},
\]
and one can check that
\begin{align}
{\mathcal S}={}&-2\big(5x^4-6x^2y^2+5y^4\big)H-2xy\big(3x^4-10x^2y^2+3y^4\big) P_x P_y\nonumber \\
&-\big(x^6-15x^4y^2+15x^2y^4-y^6\big)Q.\label{casn3}
\end{align}
As shown in Appendix~\ref{appendixA}, the corresponding Liouville metric is of constant negative curvature and should not be considered any longer.

\subsubsection[General case: n >= 4]{General case: $\boldsymbol{ n\geq 4}$}
This time we have
\begin{gather*}
A_{-2}=4\big(2C^n_2 xy^{2n-5}+\cdots+ (n-2)(n-1)
C^n_{n-1} x^{2n-5}y\big), \\
B_{-2}=4\big((n-2)(n-1)C^n_0 xy^{2n-5}+\cdots+2C^n_{n-3} x^{2n-5}y\big),\end{gather*}
as well as
\[
A_0=2\big(3C^n_1 xy^{2n-3}+\cdots+(n-1)(2n-1) C^n_{n-1} x^{2n-3}y\big)
\]
and
\begin{gather*}
A_2=4\big(C^n_0 xy^{2n-1}+\cdots+n^2C^n_{n-1} x^{2n-1}y\big),\\
B_2=4\big(n^2C^n_0 xy^{2n-1}+\cdots+C^n_{n-1}x^{2n-1}y\big).\end{gather*}
It follows that $a_0+b_0=0$ and $a_{-2}=b_{-2}=a_2=b_2=0$ hence $f(x)+g(y)\equiv 0$.
\end{proof}

\section{Global structure}\label{global}
One should take care of the following point: in principle $(H, Q, {\mathcal S})$ are independent integrals which ensure the superintegrability, but for particular values of the parameters a Killing vector may appear (let us call $K$ the corresponding linearly conserved quantity) which may induce a~reducibility either of $Q$ or of
${\mathcal S}$. So we may have, in these particular cases, a SI system either of the form $(H, Q, K)$ or $(H, {\mathcal S}, K)$.

\subsection{Case I}\label{case1}
In Proposition \ref{p3}, we got the Liouville metric
\begin{gather*}
f(x)=\mu x^2+a_1x+a_0,\qquad g(y)=\mu y^2+b_1y+b_0,\\
 {\mathcal A}=-\mu x y-\frac{b_1}{2} x-\frac{a_1}{2} y,\qquad {\mathcal S}={\mathcal A} H+P_x P_y.
 \end{gather*}
If $\mu=0$ by a translation of $x$ and $y$, we can set $a_0=b_0=0$. Then, by a linear change of coordinates and a scaling of the metric, one obtains
\begin{equation}\label{K1}
g=x \big({\rm d}x^2+{\rm d}y^2\big)\quad \Longrightarrow\quad H=\frac{P_x^2+P_y^2}{x}.
\end{equation}
This geodesic flow is not globally defined because the conformal factor vanishes for $x=0$ inducing a singularity in the Gaussian curvature
$\ R=1/\big(2x^3\big)$.

This is one of the SI systems discovered by Koenigs \cite{Ko}, see (\ref{met1}).
This case is special since~$P_y$ is now conserved and the SI system is
\[ H,\quad P_y,\quad {\mathcal S}=-\frac y2 H+P_xP_y, \]
since $Q=-P_y^2$ is reducible. It is well known \cite{kk1} that there is another quadratic integral
\[{\mathcal T}=-\frac{y^2}{2}H+2P_y(yP_x-xP_y),\]
which is reducible since we have
$ H {\mathcal T}=-2{\mathcal S}^2-2P_y^4$.

When $\mu\neq 0$ a translation of $x$ and $y$ allows $a_1=b_1=0$ and a scaling of the metric allows~$\mu=1$, leaving us with
\begin{equation}\label{K3}
g=\big(s+x^2+y^2\big)\big({\rm d}x^2+{\rm d}y^2\big),\qquad (x,y)\in {\mathbb R}^2,\qquad s=a_0+b_0.
\end{equation}
This metric is Koenigs \cite{Ko}, see (\ref{met3}). It was also considered by Matveev in \cite[p.~555]{Ma}. It is globally defined iff $s>0$: since the conformal factor never vanishes we have for mani\-fold~$M\cong{\mathbb R}^2$.

This metric was advocated in \cite[p.~565]{Ma} to be an example of a geodesic flow with four independent integrals. Indeed there are 4 integrals:
\[ H,\qquad Q=(y^2+b_0) H-P_y^2,\qquad K=xP_y-yP_x,\qquad {\mathcal S}=-xy H+P_x P_y,\]
which are not functionally independent since we have
\[ {\mathcal S}^2=-(Q+a_0H)(Q-b_0H)+H K^2. \]

\subsection{Case II}\label{case2}
In Proposition \ref{p4}, we got the Liouville metric with
\[ f(x)=\mu x^2+a x+a_0,\qquad g(y)=\frac{\mu}{4}y^2+\frac{b}{y^2}+b_0, \]
and
\[ {\mathcal A}=-\frac{\mu}{2}xy^2-\frac{a}{4}y^2-b_0 x\quad\Longrightarrow\quad {\mathcal S}={\mathcal A} H+y P_x P_y+xQ. \]

If $\mu=0$ and $a\neq 0,\ b\neq 0$ a translation of $x$ reduces the metric to
\begin{equation}\label{varK}
g=\big(b+(s+ax)y^2\big)\frac{\big({\rm d}x^2+{\rm d}y^2\big)}{y^2},\qquad x \in {\mathbb R},\quad y>0,
\end{equation}
which is never globally defined due to the zero of the conformal factor.

The three independent integrals are
\[ H=\frac{y^2}{b+(s+ax)y^2}\big(P_x^2+P_y^2\big), \qquad \widetilde{Q}=\frac{\big(b+sy^2\big)P_x^2-axy^2 P_y^2}{b+(s+ax)y^2},\]
and
\[ {\mathcal S}=-\left(\frac a4y^2+sx\right)H+y P_x P_y+x \widetilde{Q}.
\]
A companion case is obtained via $(x \leftrightarrow y).$

For $a=0$, taking $b=s^2$ and scaling the metric gives
\begin{equation}\label{K2}
g=\big(s+y^2\big)\frac{{\rm d}x^2+{\rm d}y^2}{y^2}=\big(s+y^2\big)g_0\big(H^2\big),\qquad
 x\in{\mathbb R},\quad y>0,
\end{equation}
we recover a second metric due to Koenigs, see (\ref{met2}), globally defined on $M\cong{\mathbb H}^2$ iff $s>0$. Its three independent integrals are
\[ H=\frac{y^2}{s+y^2}\big(P_x^2+P_y^2\big),\qquad P_x,\qquad {\mathcal S}=-x H+P_x(xP_x+yP_y),
\]
since $Q=P_x^2-s H$ is reducible.

If $\mu\neq 0$ we can set $\mu=1$. For $a=b=0$, one gets
\begin{equation}\label{M}
g=\left(s+x^2+\frac{y^2}{4}\right)\big({\rm d}x^2+{\rm d}y^2\big),
\end{equation}
a metric derived by Matveev \cite{Ma}, globally defined on $M\cong {\mathbb R}^2$ iff $s>0$. Its three independent integrals are
\[ H=\frac{P_x^2+P_y^2}{ s+x^2+\frac{y^2}{4}},\qquad \widetilde{Q}=P_x^2-x^2 H,\qquad {\mathcal S}=-x\left(s+\frac{y^2}{2}\right)H+y P_x P_y
+x \widetilde{Q}.
\]

More generally, for $b\neq 0$, we have
\begin{equation}\label{SupM}
g=\Psi(x,y) g_0\big(H^2\big),\qquad \Psi(x,y)=b+y^2\left(s+x^2+\frac{y^2}{4}\right),
\end{equation}
globally defined on $M\cong{\mathbb H}^2$ iff $b>0$ and $s\geq 0$.

The three independent integrals are
\[ H=\frac{y^2}{\Psi}\big(P_x^2+P_y^2\big),\qquad \widetilde{Q}=P_x^2-x^2 H, \qquad
{\mathcal S}=-x\left(s+\frac{y^2}{2}\right)H+y P_x P_y+x \widetilde{Q}.\]
Due to Remark~\ref{remark4.4}, Case~III does not produce new metrics and can be skipped.

\subsection{Case IV}\label{case4}
In Proposition \ref{p6}, we got the Liouville metric
\[ f(x)=\mu x^2+\frac a{x^2}+a_0,\qquad g(y)=\mu y^2+\frac{b}{y^2}+b_0, \]
and
\begin{equation}\label{IVS}
{\mathcal A}=-\mu x^2y^2-\frac{b_0}{2} x^2-\frac{a_0}{2} y^2,\qquad {\mathcal S}={\mathcal A} H+xy P_x P_y+\frac 12\big(x^2-y^2\big) Q.
\end{equation}
If $\mu=0$, the case where $a_0=b_0=0$ is special since the metric is
\begin{equation}\label{cas41}
g=\frac{bx^2+ay^2}{x^2y^2}\big({\rm d}x^2+{\rm d}y^2\big),\qquad a>0,\quad b>0,\quad a\neq b, \end{equation}
encountered in Section~\ref{7p1} for $n=2$. It exhibits an extra Killing vector: $x\partial_x+y\partial_y$.

We have for integrals
\begin{equation}\label{ham4}
H=\frac{x^2y^2}{bx^2+ay^2}\big(P_x^2+P_y^2\big),\qquad Q=P_x^2-\frac a{x^2} H,\qquad K=xP_x+yP_y.
\end{equation}
However, there are two other quadratic integrals: the first one from (\ref{IVS}) is reducible
\[ 2{\mathcal S}=2xy P_x P_y+\big(x^2-y^2\big)Q=K^2-(a+b)H.
\]
So we remain with 4 integrals: $H$, $Q$, $K$ and ${\mathcal S}$ given by (\ref{casn2}) which are not independent since one has the relation
\[ Q {\mathcal S}=-(a-b)^2 H^2+2(a+b) H K^2-K^4,
\]
and we remain with a SI geodesic flow with three independent integrals: $(H, Q, K)$.

More generally, when $\mu\geq 0$ and $s=a_0+b_0\geq 0$ the metric is
\begin{equation}\label{cas42}
g=\left(\frac{bx^2+ay^2}{x^2+y^2}+s\frac{x^2y^2}{x^2+y^2}+\mu x^2y^2\right)\frac{x^2+y^2}{x^2y^2}\big({\rm d}x^2+{\rm d}y^2\big), \qquad x>0,\quad y>0,
\end{equation}
and taking into account Proposition \ref{noncan}, we have $g=\chi g_0\big(H^2\big).$

In terms of the global coordinates $\big(X^1, X^2, X^3\big)$, defined in Appendix~\ref{appendixA}, relation (\ref{coord}), we have
\[ x^2+y^2=\sqrt{\frac{X^3+X^1}{X^3-X^1}},\qquad x^2-y^2=\frac{X^2}{X^3-X^1},\]
leading to
\[ \frac{bx^2+ay^2}{x^2+y^2}=\frac{a+b}{2}-\frac{(a-b)}{2}\frac{X^2}{\sqrt{1+(X^2)^2}}.\]
Since $ \frac{X^2}{\sqrt{1+(X^2)^2}} \in (-1,+1)$ the sum of the first two terms, for $a>0$ and $b>0$ is strictly positive, so if $(a>0$, $b>0)$ and $(s\geq 0$, $\mu\geq 0)$ the non-vanishing of the conformal factor implies~$M\cong{\mathbb H}^2$.

The three independent integrals are
\[ H=\frac{x^2y^2}{x^2+y^2}\frac{\big(P_x^2+P_y^2\big)}{\chi},\qquad \widetilde{Q}=P_x^2-\left(\frac a{x^2}+\mu x^2\right) H,\]
and
\[ {\mathcal S}=-x^2\left(\frac s2+\mu y^2\right) H+xy P_x P_y+\frac 12\big(x^2-y^2\big)\widetilde{Q}.\]

\begin{Remark}
The global structure meets the manifolds ${\mathbb R}^2$ and ${\mathbb H}^2$ but never ${\mathbb S}^2$. The explanation stems from a theorem due to Kiyohara \cite{Ki}: a SI geodesic flow of Hamiltonian $H$, globally defined on ${\mathbb S}^2$, with two extra quadratic integrals implies that its metric is of constant curvature hence cannot appear in our analysis. This applies as well to Koenigs metrics which are never defined on ${\mathbb S}^2$.
\end{Remark}

Let us now relate our results, via coupling constant metamorphosis, to previous work.

\section{Coupling constant metamorphosis}\label{section6}
The so-called ``coupling constant metamorphosis" \cite{hgdr} or St\"ackel transform \cite{bkm} establishes that if a system has for Hamiltonian
\begin{equation}\label{St}
H=P_x^2+P_y^2 +V(x,y)
\end{equation}
and is quadratically SI then
\[ H'=\frac{P_x^2+P_y^2}{V(x,y)}\]
will be also quadratically SI. Since all the systems having the form (\ref{St}) were derived in \cite{kkmp} let us give their relation with our work.

In Case~I (Section~\ref{case1}), we have
\[V=\mu\big(x^2+y^2\big)+2(a_1 x+b_1 y)+s,\qquad s=a_0+b_0,\]
which is the case ${\bf E'3}$, merely the case ${\bf E3}$ with a translation of both variables.

In Case~II (Section~\ref{case2}), the potential is
\[V=\mu\Big(x^2+\frac{y^2}{4}\Big)+2a_1x+\frac{b}{y^2}+s,\]
which is the case ${\bf E2}$. Case~III is obtained by the permutation $x \leftrightarrow y$.

In Case IV, the potential is
\[ V=\mu\big(x^2+y^2\big)+\frac{a}{x^2}+\frac{b}{y^2}+s,\]
and this is the case ${\bf E1}$.

\appendix

\section{Appendix}\label{appendixA}
Let us prove:

\begin{Proposition}\label{noncan} If one takes
\[ f(x)=\frac 1{x^2},\qquad g(y)=\frac 1{y^2}\quad\Longrightarrow\quad g_0=\frac{\big(x^2+y^2\big)}{x^2y^2}\big({\rm d}x^2+{\rm d}y^2\big),\qquad x>0,\quad y>0,\]
the Liouville metric is a non-canonical form of $g_0\big(H^2\big)$.
\end{Proposition}

\begin{proof} The Gaussian curvature $R=-1$ implies a non-canonical metric of ${\mathbb H}^2$. There are three linear integrals
\[ K_1=xP_x+yP_y,\qquad K_2=\frac{xP_x-yP_y}{x^2+y^2},\qquad K_3=x\big(x^2-3y^2\big)P_x-y\big(y^2-3x^2\big)P_y.
\]
Within this geometry all the quadratic integrals are reducible:
\[
H=\frac 14\big(K_1^2-K_2 K_3\big),\qquad Q=K_1 K_2,
\]
and despite the apparent complexity of ${\mathcal S}$, see (\ref{casn3}), we have merely $\ {\mathcal S}=-K_3^2$.

Defining the coordinates
\begin{equation}\label{coord}
X^1=\frac{\big(x^2+y^2\big)^2-1}{4xy}\in {\mathbb R},\qquad X^2=\frac{x^2-y^2}{2xy}\in {\mathbb R},\qquad X^3=\frac{\big(x^2+y^2\big)^2+1}{4xy} \geq 1,
\end{equation}
one can check the relations
\[ \big(X^1\big)^2+\big(X^2\big)^2-\big(X^3\big)^2=-1 \]
and
\[ g_0\big(H^2\big)\equiv \big({\rm d}X^1\big)^2+\big({\rm d}X^2\big)^2-\big({\rm d}X^3\big)^2=\frac{\big(x^2+y^2\big)}{x^2y^2}\big({\rm d}x^2+{\rm d}y^2\big).\tag*{\qed}\]\renewcommand{\qed}{}
\end{proof}

\section{Appendix}\label{appendixB}
The list of SI Liouville metrics is the following:
\begin{enumerate}\itemsep=0pt
\item The first Koenigs metric (\ref{K1}):
\[ g=x\big({\rm d}x^2+{\rm d}y\big)^2,\qquad x>0,\quad y\in{\mathbb R},\quad M\cong \varnothing.\]
\item The second Koenigs metric (\ref{K2}):
\[ g=\big(s+x^2\big)\frac{{\rm d}x^2+{\rm d}y^2}{x^2}=\big(s+x^2\big) g_0\big(H^2\big), \qquad x>0,\quad y\in {\mathbb R},\qquad s>0 \ \Longrightarrow\ M\cong {\mathbb H}^2.\]
\item The third Koenigs metric (\ref{K3}):
\[ g=\big(s+x^2+y^2\big)\big({\rm d}x^2+{\rm d}y^2\big),\qquad (x,y)\in {\mathbb R}^2,\qquad s>0 \quad \Longrightarrow\quad M\cong{\mathbb R}^2.\]
\item The variant (\ref{varK}):
\[g=\big(b+(s+ax)y^2\big) g_0\big(H^2\big),\qquad x \in {\mathbb R},\quad y>0,\qquad M=\varnothing,\]
and its companion via $(x \leftrightarrow y)$.
\item Matveev metric (\ref{M}):
\[ g=\left(s+x^2+\frac{y^2}{4}\right)\big({\rm d}x^2+{\rm d}y^2\big),\qquad (x,y)\in {\mathbb R^2},\qquad s>0 \quad \Longrightarrow\quad M\cong{\mathbb R}^2,\]
and its companion via $(x \leftrightarrow y)$.
\item The most general metric in Case~II (see (\ref{SupM})):
\begin{gather*}
 g=\left[b+y^2\left(s+x^2+\frac{y^2}{4}\right)\right] g_0\big(H^2\big),\qquad x\in {\mathbb R},\quad y>0, \qquad b>0,\quad s\geq 0\\
 \Longrightarrow\quad M\cong{\mathbb H^2},
\end{gather*}
and its companion via $(x \leftrightarrow y)$.
\item The metric (\ref{cas41}):
\begin{gather*}
 g=\frac{bx^2+ay^2}{x^2y^2}\big({\rm d}x^2+{\rm d}y^2\big),\qquad a>0,\quad b>0,\quad a\neq b\quad
\Longleftrightarrow\quad M\cong {\mathbb H}^2.
\end{gather*}
For $a=b$, it reduces to a non-canonical metric on
${\mathbb H}^2$, see Appendix~\ref{appendixA}.
\item The most general metric in Case~IV (Section~\ref{case4}) given by (\ref{cas42})
\[ g=\chi(x,y) g_0\big(H^2\big),\qquad \chi(x,y)=\frac{bx^2+ay^2}{x^2+y^2}+s\frac{x^2y^2}{x^2+y^2}+\mu x^2y^2,\]
and if $a>0$, $b>0$, $a\neq b$, $s\geq 0$, $\mu\geq 0$, we have $M\cong {\mathbb H}^2.$
\end{enumerate}

The three independent integrals are given for each metric in Section \ref{global}.

\section{Appendix}\label{appendixC}
The metrics derived in \cite{bmm} are expressed in terms of algebraic functions while Koenigs ones resort to trigonometric or hyperbolic functions. Let us put a bridge, in the Riemannian case, between these metrics observing that a scaling of the metric is irrelevant since this is tantamount to a~scaling of the Hamiltonian.

\subsection{The affine case}
The metric (a) given in \cite{bmm}, in which we substitute $u={\rm e}^x$, leads to
\begin{equation}\label{aff1}
u\big({\rm d}u^2+{\rm d}y^2\big).
\end{equation}
The metric (b) given in \cite{bmm}, in which we substitute $u=2\sqrt{{\rm e}^x+\epsilon}$, becomes
\begin{equation}\label{aff2}
\big(u^2-4\epsilon\big)\frac{\big({\rm d}u^2+{\rm d}y^2\big)}{u^2}.
\end{equation}
Let us compare with the affine case in (\ref{Kmain})
\[ g=\frac{\big(a_2 x^2+2a_1 x+a_0\big)}{(a_2 x+a_1)^2} \big({\rm d}x^2+{\rm d}y^2\big).\]
If $a_2=0$, we can set $a_1=1/2$ and taking $u=x+a_0$ we obtain (\ref{aff1}). If $a_2$ does not vanish, we take $a_2=1$. Defining $u=x+a_1$, we recover (\ref{aff2}) up to an overall scaling of $u$ and $y$.

What remains to be discussed in \cite{bmm} is the metric (c) which can be written
\[ g=\frac x{\big(x^2+2ax+\epsilon\big)}\left(\frac{{\rm d}x^2}{x^2\big(x^2+2ax+\epsilon\big)}+{\rm d}y^2\right),\qquad a\in{\mathbb R},\qquad \epsilon=\pm 1.\]

\subsection{The hyperbolic case}
If $\epsilon=-1$, let us define the coordinate change $x \to u $:
\[ u=\arctan\left(\frac{\sqrt{x^2+2ax-1}}{1-ax}\right),\qquad a<0, \qquad x>x_0=\sqrt{a^2+1}+|a|.\]
It maps $x\in (x_0,+\infty)$ into $u\in (0,\arctan(1/|a|))$ and we have the relations
\[ \frac{{\rm d}x^2}{x^2\big(x^2+2ax-1\big)}={\rm d}u^2, \qquad x=\frac 1{\sqrt{a^2+1}\cos u-|a|},\]
leading to
\[ g=\frac 1{a^2+1}\frac{\big(\sqrt{a^2+1}\cos u-|a|\big)}{\sin^2 u}\big({\rm d}u^2+{\rm d}y^2\big), \]
which fits with the first metric in (\ref{Kmain}).

\subsection{The trigonometric cases}
For $\epsilon=+1$, we have two possible cases. The first one defines the coordinate change
\[ u={\rm arctanh}\left(\frac{\sqrt{x^2+2ax+1}}{1+ax}\right),\qquad a>1, \quad x>0,\]
which maps $x\in(0,+\infty)$ into $u\in({\rm arctanh}(1/a),+\infty)$. We have the relations
\[ \frac{{\rm d}x^2}{x^2\big(x^2+2ax+1\big)}={\rm d}u^2, \qquad x=\frac 1{\sqrt{a^2-1}\cosh u-a},\]
leading to
\[ g=\frac 1{\big(a^2-1\big)}\frac{\sqrt{a^2-1}\cosh u-a}{\sinh^2 u}\big({\rm d}u^2+{\rm d}y^2\big), \]
which fits with the fourth metric in (\ref{Kmain}).

The second case is given by the coordinate change
\[ u={\rm arctanh}\left(\frac{1+ax}{\sqrt{x^2+2ax+1}}\right),\qquad a\in (0,1), \quad x>0,\]
which maps $x>0$ into $u\in ({\rm arctanh}(a),+\infty)$. We have
\[ \frac{{\rm d}x^2}{x^2\big(x^2+2ax+1\big)}={\rm d}u^2, \qquad x=\frac 1{\sqrt{1-a^2}\sinh u-a},\]
leading to
\[ g=\frac 1{1-a^2}\frac{\sqrt{1-a^2}\sinh u-a}{\cosh^2 u}\big({\rm d}u^2+{\rm d}y^2\big),\]
which fits with the third metric in (\ref{Kmain}).

Up to now $a\neq \pm 1$. For $a=1$, we have
\[x=\frac{{\rm e}^{-u}}{{\rm e}^{-u}-1} \quad\Rightarrow\quad g=\big({\rm e}^{-u}-{\rm e}^{-2u}\big)\big({\rm d}u^2+{\rm d}y^2\big).\]
For $a=-1$, we have two possible cases:
\begin{gather*}
 x=\frac{{\rm e}^u}{{\rm e}^u-1} \quad\Rightarrow\quad g=\big({\rm e}^{2u}-{\rm e}^u\big)\big({\rm d}u^2+{\rm d}y^2\big),\\
 x=\frac{{\rm e}^u}{{\rm e}^u+1} \quad\Rightarrow\quad g=\big({\rm e}^u+{\rm e}^{2u}\big)\big({\rm d}u^2+{\rm d}y^2\big).
 \end{gather*}
Due to the freedom $u \to \pm u$, we are in agreement with the second case in (\ref{Kmain}).

\subsection*{Acknowledgements}

The author is greatly indebted to the anonymous referees whose remarks allowed useful corrections and put to light several references related with this work.

\pdfbookmark[1]{References}{ref}
\LastPageEnding


\begin{thebibliography}{99}
\footnotesize\itemsep=0pt

\bibitem{behr}
Ballesteros A., Enciso A., Herranz F.J., Ragnisco O., Riglioni D., Quantum
 mechanics on spaces of nonconstant curvature: the oscillator problem and
 superintegrability, \href{https://doi.org/10.1016/j.aop.2011.03.002}{\textit{Ann. Physics}} \textbf{326} (2011), 2053--2073,
 \href{https://arxiv.org/abs/1102.5494}{arXiv:1102.5494}.

\bibitem{bmp}
Bolsinov A.V., Matveev V.S., Pucacco G., Normal forms for pseudo-{R}iemannian
 2-dimensional metrics whose geodesic flows admit integrals quadratic in
 momenta, \href{https://doi.org/10.1016/j.geomphys.2009.04.010}{\textit{J.~Geom. Phys.}} \textbf{59} (2009), 1048--1062,
 \href{https://arxiv.org/abs/0803.0289}{arXiv:0803.0289}.

\bibitem{bkm}
Boyer C.P., Kalnins E.G., Miller Jr. W., St\"{a}ckel-equivalent integrable
 {H}amiltonian systems, \href{https://doi.org/10.1137/0517057}{\textit{SIAM~J. Math. Anal.}} \textbf{17} (1986),
 778--797.

\bibitem{bmm}
Bryant R.L., Manno G., Matveev V.S., A solution of a problem of {S}ophus {L}ie:
 normal forms of two-dimensional metrics admitting two projective vector
 fields, \href{https://doi.org/10.1007/s00208-007-0158-3}{\textit{Math. Ann.}} \textbf{340} (2008), 437--463,
 \href{https://arxiv.org/abs/0705.3592}{arXiv:0705.3592}.

\bibitem{dy}
Daskaloyannis C., Ypsilantis K., Unified treatment and classification of
 superintegrable systems with integrals quadratic in momenta on a
 two-dimensional manifold, \href{https://doi.org/10.1063/1.2192967}{\textit{J.~Math. Phys.}} \textbf{47} (2006), 042904,
 38~pages, \href{https://arxiv.org/abs/math-ph/0412055}{arXiv:math-ph/0412055}.

\bibitem{Fo}
Fordy A.P., First integrals from conformal symmetries: {D}arboux--{K}oenigs
 metrics and beyond, \href{https://doi.org/10.1016/j.geomphys.2019.07.006}{\textit{J.~Geom. Phys.}} \textbf{145} (2019), 103475,
 13~pages, \href{https://arxiv.org/abs/1804.06904}{arXiv:1804.06904}.

\bibitem{fh}
Fordy A.P., Huang Q., Generalised {D}arboux--{K}oenigs metrics and
 3-dimensional superintegrable systems, \href{https://doi.org/10.3842/SIGMA.2019.037}{\textit{SIGMA}} \textbf{15} (2019),
 037, 30~pages, \href{https://arxiv.org/abs/1810.13368}{arXiv:1810.13368}.

\bibitem{hgdr}
Hietarinta J., Grammaticos B., Dorizzi B., Ramani A., Coupling-constant
 metamorphosis and duality between integrable {H}amiltonian systems,
 \href{https://doi.org/10.1103/PhysRevLett.53.1707}{\textit{Phys. Rev. Lett.}} \textbf{53} (1984), 1707--1710.

\bibitem{kk2}
Kalnins E.G., Kress J.M., Miller Jr. W., Winternitz P., Superintegrable systems
 in {D}arboux spaces, \href{https://doi.org/10.1063/1.1619580}{\textit{J.~Math. Phys.}} \textbf{44} (2003), 5811--5848,
 \href{https://arxiv.org/abs/math-ph/0307039}{arXiv:math-ph/0307039}.

\bibitem{kkmp}
Kalnins E.G., Kress J.M., Pogosyan G.S., Miller Jr. W., Completeness of
 superintegrability in two-dimensional constant-curvature spaces,
 \href{https://doi.org/10.1088/0305-4470/34/22/311}{\textit{J.~Phys.~A}} \textbf{34} (2001), 4705--4720, \href{https://arxiv.org/abs/math-ph/0102006}{arXiv:math-ph/0102006}.

\bibitem{kk1}
Kalnins E.G., Kress J.M., Winternitz P., Superintegrability in a
 two-dimensional space of nonconstant curvature, \href{https://doi.org/10.1063/1.1429322}{\textit{J.~Math. Phys.}}
 \textbf{43} (2002), 970--983, \href{https://arxiv.org/abs/math-ph/0108015}{arXiv:math-ph/0108015}.

\bibitem{Ki}
Kiyohara K., Compact {L}iouville surfaces, \href{https://doi.org/10.2969/jmsj/04330555}{\textit{J.~Math. Soc. Japan}}
 \textbf{43} (1991), 555--591.

\bibitem{Ko}
Koenigs G., Sur les g\'eod\'esiques \`a int\'egrales quadratiques, in
 Le\c{c}ons sur la th\'eorie g\'en\'erale des surfaces, Vol.~4, Editor
 J.G.~Darboux, Chelsea Publishing, New York, 1972, 368--404.

\bibitem{Ma}
Matveev V.S., Lichnerowicz--{O}bata conjecture in dimension two,
 \href{https://doi.org/10.4171/CMH/25}{\textit{Comment. Math. Helv.}} \textbf{80} (2005), 541--570.

\bibitem{ms}
Matveev V.S., Shevchishin V.V., Two-dimensional superintegrable metrics with
 one linear and one cubic integral, \href{https://doi.org/10.1016/j.geomphys.2011.02.012}{\textit{J.~Geom. Phys.}} \textbf{61}
 (2011), 1353--1377, \href{https://arxiv.org/abs/1010.4699}{arXiv:1010.4699}.

\bibitem{mpw}
Miller Jr. W., Post S., Winternitz P., Classical and quantum superintegrability
 with applications, \href{https://doi.org/10.1088/1751-8113/46/42/423001}{\textit{J.~Phys.~A}} \textbf{46} (2013), 423001, 97~pages,
 \href{https://arxiv.org/abs/1309.2694}{arXiv:1309.2694}.

\bibitem{Va1}
Valent G., Global structure and geodesics for {K}oenigs superintegrable
 systems, \href{https://doi.org/10.1134/S1560354716050014}{\textit{Regul. Chaotic Dyn.}} \textbf{21} (2016), 477--509,
 \href{https://arxiv.org/abs/1510.08379}{arXiv:1510.08379}.

\bibitem{vds}
Valent G., Duval C., Shevchishin V., Explicit metrics for a class of
 two-dimensional cubically superintegrable systems, \href{https://doi.org/10.1016/j.geomphys.2014.08.004}{\textit{J.~Geom. Phys.}}
 \textbf{87} (2015), 461--481, \href{https://arxiv.org/abs/1403.0422}{arXiv:1403.0422}.

\end{thebibliography}
\end{document}